\documentclass{SCIS2020}
\usepackage{cite}
\usepackage{graphicx}
\usepackage{amsmath}
\usepackage{times}
\usepackage{latexsym}
\usepackage{bm}
\usepackage{amssymb}
\usepackage{amsfonts}

\begin{document}
\ArticleType{RESEARCH PAPER}
\Year{2020}
\Month{}
\Vol{}
\No{}
\DOI{}
\ArtNo{}
\ReceiveDate{}
\ReviseDate{}
\AcceptDate{}
\OnlineDate{}

\title{Statistical CSI based Design for Intelligent Reflecting Surface Assisted MISO Systems}{Statistical CSI based Design for Intelligent Reflecting Surface Assisted MISO Systems}

\author[1]{Xiaoling Hu}{}
\author[1]{Junwei Wang}{}
\author[1]{Caijun Zhong}{{caijunzhong@zju.edu.cn}}

\AuthorMark{Hu X L}

\AuthorCitation{ Hu X L, Wang J W, Zhong C J}


\address[1]{College of information Science and electronic engineering, Zhejiang University, Hangzhou {\rm 310027}, China}

\abstract{
This paper considers an intelligent reflecting surface (IRS) aided multiple-input single-output communication system, where statistical channel state information (CSI) is exploited for transmit beamforming and IRS beamforming. A tight upper bound is derived for the ergodic capacity of the system. Based on which, the joint optimization of transmit beam and IRS beam are studied. Depending on whether a line-of-sight path exists between the access point and user, two different cases, namely, Rician fading and  Rayleigh fading, are separately treated. Specifically, for the Rician fading case, an iterative algorithm is proposed, which is guaranteed to converge. For the Rayleigh fading case, closed-form designs are obtained for the transmit beam and IRS beam. Simulation results show the proposed beamforming scheme achieves similar performance as the benchmark algorithm requiring instantaneous CSI.}


\maketitle

\section{Introduction}

Intelligent reflecting surfaces (IRS) have
recently emerged as a promising low-cost approach to enhance the spectral efficiency and extend the coverage of communication systems \cite{wu2019towards}. Specifically, the IRS enables the manipulation of signal propagation environment by smartly controlling a large number of low-cost, passive reflecting elements, each of which can independently reflect the incident signal with an adjustable phase shift \cite{cui2014coding}.

Due to the potential of realizing a software programmable propagation environment, IRS has attracted enormous attentions from the research community. The fundamental performance of IRS-assisted communications has been studied in \cite{Q.Tao1}, while beamforming design in IRS-aided communication systems has been considered in \cite{hu2020programmable,AI,wu2019intelligent, wu2019beamforming,yang2020irs}. Specifically, in \cite{hu2020programmable}, a low-complexity beamforming method based on beam training is proposed. Capitalizing on the recent advance in artificial intelligence \cite{AI0}, an unsupervised learning approach is proposed to design the passive beamforming at the IRS in \cite{AI}. Furthermore, the work \cite{wu2019intelligent} considered the joint optimization of active beamforming at the access ponit (AP) and passive beamforming at the IRS. Aiming at minimizing the transmit power subject to individual rate constraints, the optimization problem is solved by the alternating optimization approach. Later on, taking into account of the practical constraint of discrete phase shift, a mixed-integer non-linear program problem is formulated to jointly optimize the transmit beam and phase shifts \cite{wu2019beamforming}. In addition, the integration of IRS with other transmission technologies has been investigated, such as orthogonal frequency division multiple access (OFDMA) \cite{yang2020irs}, non-orthogonal multiple access (NOMA)\cite{NOMA}, physical layer security \cite{PLS}, two-way relaying \cite{twoway} and millimeter wave communications \cite{mmwave}.

To fully realize the potential of IRS, the availability of channel state information (CSI) is crucial \cite{J.Zhang}. However, due to the passive architecture of the IRS as well as its massive number of reflecting elements, the acquisition of CSI is a challenging task. In general, the reported channel estimation methods can be divided into two categories. One is to separately estimate the direct channel from the AP to the IRS and the IRS-user channel \cite{taha2019enabling}, by assuming that the IRS can either work in the reflection mode for data transmission or receive mode for channel estimation. However, the realization of receiving mode requires a large number of receive radio frequency (RF) chains, which would significantly increase hardware cost as well as power consumption. Another typical approach is to estimate the cascaded channel at the AP, by pilot training and proper design of IRS reflection patterns \cite{yang2020intelligent,zheng2019intelligent,mishra2019channel}. However, since the required pilot length is proportional to the number of reflecting elements, the training overhead would become prohibitive as the number of reflecting elements grows.

To overcome the drawbacks of above channel estimation approaches, one promising direction is to use only statistical CSI for beamforming design. The reason is that, the statistical CSI is relatively easy to obtain through long-term observation. In addition, the statistical CSI varies slowly, hence no frequent update is required, which substantially reduces the training overhead. In a recent work \cite{shijin}, the authors proposed to adopt statistical CSI for the design of phase shifts at the RIS. However, they assume that the AP uses the maximum ratio transmission, which requires the instantaneous CSI of the composite channel. Motivated by this, this paper considers a more practical scenario where only statistical CSI is available at the AP, and pursue a joint design of the transmit beamformer and phase shifts to maximize the ergodic capacity. Depending on whether a line-of-sight (LOS) path exists between the AP and the user, two separate cases are addressed. The main contributions of this paper are summarized as follows:
 \begin{itemize}
     \item In case there is LOS path between the AP and user, Rician fading is used to model the direct channel. Then, an alternating algorithm is proposed to maximize the ergodic capacity by jointly designing the transmit beam at the AP and phase shifts at the IRS. Specifically, the joint optimization problem is tackled by solving two sub-problems in an iterative manner. For each sub-problem, the optimal solution is obtained in closed form. Moreover, the proposed alternating algorithm is guaranteed to converge. Simulation results show that the proposed algorithm achieves similar performance as the algorithm requiring perfect and instantaneous CSI, although only statistical CSI is exploited.
     \item In case there is no LOS path between the AP and user, Rayleigh fading is used to model the direct channel. In this case, closed-form optimal solutions for both the transmit beam and the phase shifts are obtained. Simulation results show that the performance of our proposed beamforming scheme is almost identical to that of the algorithm requiring perfect and instantaneous CSI \cite{wu2019intelligent}. Moreover, since closed-form solutions can be obtained, the proposed scheme has much lower complexity when compared with the iterative algorithm in   \cite{wu2019intelligent}.
 \end{itemize}

The remainder of the paper is organized as follows. Section \ref{S2} introduces the IRS-aided communication system model. The joint optimization problem of the transmit beamforming and IRS beamforming is solved in Section \ref{S3}, while numerical simulation results are presented in Section \ref{S4}. Finally, Section \ref{S5} concludes the paper.

\emph{Notation:} Throughout the paper, boldface lower-case and upper-case letters indicate vectors and matrices, respectively. $\mathbb{C}$ is the set of complex numbers. $(\cdot)^T$ , $(\cdot)^*$ and $(\cdot)^{H}$ represent the transpose, complex conjugate and Hermitian transpose, respectively. $|\cdot|$ and $\|\cdot\|$ denote the complex modulus and vector Euclidean norm, respectively. $\mathsf{diag}\{\cdot\}$ and $\mathsf{norm}(\cdot)$ indicate the diagonalization and normalization of a vector. $\mathsf{tr}(\cdot)$ returns the trace of a matrix. $\mathbb{E}\{\cdot\}$ is the statistical expectation. $\mathsf{angle} ( \cdot )$  denotes the angle in radians.

\section{System Model} \label{S2}
We consider a three node system where the multi-antenna AP communicates with the single antenna user assisted by an IRS, as illustrated in Fig. \ref{model}. Both the transmitter and the IRS have the uniform linear array (ULA). The AP is equipped with $M$ antennas, while the IRS is equipped with $N$ reflecting elements. In addition to the IRS-assisted propagation path, the direct link between the AP and the user is taken into account. Frequency-flat channels are considered. Moreover, we assume that only statistical CSI is available at the AP, and the control link between the AP and the IRS is error free.

\begin{figure}[htbp]
	\centering
	\includegraphics[width=3in]{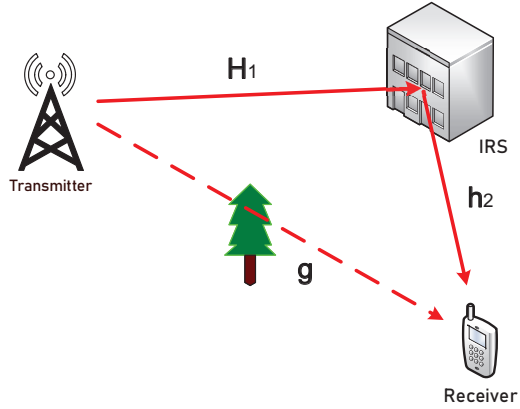}
	\caption{System model}
	\label{model}
\end{figure}

The received signal $y$ can be expressed as
\begin{equation}
y = \sqrt{P}\left(\frac{\mathbf{h}_2^T \boldsymbol{\Phi} \mathbf{H}_1}{\sqrt{d_1^{\alpha_1}d_2^{\alpha_2}}}+\frac{\mathbf{g}^T}{\sqrt{d_0^{\alpha_0}}}\right)\mathbf{f}x + z,
\end{equation}
where $P$ is the transmit power, $x$ is the Gaussian input signal satisfying $\mathbb{E}\{xx^*\}=1$, $\mathbf{f} \in \mathbb{C}^{M \times 1}$ is the transmit beamforming vector with $\|\mathbf{f}\|^2 = 1$, $\boldsymbol{\Phi}={\sf diag}\{\boldsymbol{\phi}\}$ is the phase shift matrix, where $\boldsymbol{\phi}=[\phi_1, \phi_2, \dots,\phi_N]^T\in \mathbb{C}^{N \times 1}$ and $|\phi_n|=1, n=1,\dots,N$. In addition, $\mathbf{H}_1 \in \mathbb{C}^{N \times M}$ denotes channel between the transmitter and the IRS, $\mathbf{h}_2 \in \mathbb{C}^{N \times 1}$ denotes the channel between the IRS and the receiver, $\mathbf{g}\in \mathbb{C}^{M \times 1}$ denotes the channel between the transmitter and the receiver, and $z$ is the additive white Gaussian noise (AWGN) with zero mean and variance $N_0$. Moreover, $d_k$ and $\alpha_k$, $k=0,1,2$ denote the distance between the nodes, and pathloss exponent of the channels, respectively.

We assume that the IRS is deployed at a desirable location, such that there exists both line-of-sight (LOS) path to the AP and user. Hence, Rician distribution is adopted to model the fading channels, i.e.,
\begin{align}
\mathbf{H}_1 = \sqrt{\frac{K_1}{K_1+1}} \mathbf{\bar{H}}_1 + \sqrt{\frac{1}{K_1+1}} \mathbf{\tilde{H}}_1,
\end{align}
and
\begin{align}
\mathbf{h}_2 = \sqrt{\frac{K_2}{K_2+1}} \mathbf{\bar{h}}_2 + \sqrt{\frac{K_2}{K_2+1}}\mathbf{\tilde{h}}_2,
\end{align}
where $K_i$, $i=1,2$ denotes the Rician $K$ factor, $\mathbf{\tilde{H}}_1$ and $\mathbf{\tilde{h}}_2$ denotes the non-LoS components, whose elements follow the zero mean complex Gaussian distribution with unit variance. Also, $\mathbf{\bar{H}}_1$ and $\mathbf{\bar{h}}_2$ denote the LoS components, which are given by the response of the uniform linear array (ULA). For instance,
\begin{equation}\label{LoS}
\mathbf{\bar{H}}_1 = \mathbf{a}_N(\theta_\text{AoA,1}) \mathbf{a}_M^T(\theta_\text{AoD,1}), \mbox{ and } \mathbf{\bar{h}}_2 = \mathbf{a}_N(\theta_\text{AoD,2}),
\end{equation}
where $\mathbf{a}_N(\theta) \triangleq [1, e^{j\theta}, \dots, e^{j(N-1)\theta}]^T$, $\theta_\text{AoA,1}$ denotes the angle of departure arrival (AoA) at the IRS, $\theta_\text{AoD,1}$ and $\theta_\text{AoD,2}$ denote the angles of departure departure (AoD) from the AP and IRS, respectively.

As for the direct channel between the AP and user, we consider two separate cases, namely Rician fading and Rayleigh fading. Since Rayleigh fading is a special case of Rician fading, the channel vector can be expressed as
\begin{align}
\mathbf{g} = \sqrt{\frac{K_0}{K_0+1}} \mathbf{\bar{g}} +\sqrt{\frac{1}{K_0+1}}\mathbf{\tilde{g}},
\end{align}
where $K_0$ is the Rician $K$ factor, $\mathbf{\tilde{g}}$ denotes the non-LoS components, whose elements follow the zero mean complex Gaussian distribution with unit variance. Also, $\mathbf{\bar{g}}\triangleq \mathbf{a}_M(\theta_\text{AoD,0})$ denotes the LoS components.

With CSI at the receiver, the instantaneous channel capacity of the system can be expressed as
\begin{align*}
C = \log_2 \left( 1+\gamma_0\left|(\mathbf{h}_2^T \boldsymbol{\Phi} \mathbf{H}_1+ \lambda\mathbf{g}^T) \mathbf{f} \right|^2 \right),
\end{align*}
where $\gamma_0=\frac{P}{d_1^{\alpha_1}d_2^{\alpha_2}N_0}$ and $\lambda = \sqrt{\frac{d_1^{\alpha_1}d_2^{\alpha_2}}{d_0^{\alpha_0}}}$.

We consider the realistic scenario that the AP and IRS only have access to the statistical CSI, as such, instead of maximizing the the instantaneous channel capacity, our objective is to maximize the ergodic capacity by joint design of the beamforming vector $\mathbf{f}$ and the phase shift beam $\boldsymbol{\phi}$. Therefore, we have the following optimization problem P1:
\begin{equation}
\label{problem1}
\begin{aligned}
\max_{\mathbf{f},  {\bm \phi}} \quad & \mathbb{E} \left\{ C \right\} \\
\mathrm{s.t.} \quad & \|\mathbf{f}\|^2 = 1 \\
& | {\phi}_i| = 1, i=1, \dots, N.
\end{aligned}
\end{equation}

\section{Joint Design of Transmit Beam and Phase Shift Beam} \label{S3}
In this section, we focus on solving optimization problem P1. Due to the involved expression, it is challenging to derive the exact expression for the ergodic capacity. Motivated by this, we first present a tight and tractable upper bound for the ergodic capacity, and then propose efficient algorithms to maximize the upper bound by joint design of the transmit beam and phase shifts.

\begin{proposition}
	\label{proposition}
	The ergodic capacity of the system is upper bounded by
	\begin{equation}
	\mathbb{E} \left\{ C \right\}\leq C_\mathrm{up} = \log_2\left(1 + \gamma_0 \left(  \left|\left(a_2 a_1 \mathbf{\bar{h}}_2^T \boldsymbol{\Phi} \mathbf{\bar{H}}_1 + \lambda a_0 \mathbf{\bar{g}}^T\right) \mathbf{f}\right|^2
	+ b_2^2 a_1^2 \left\|\mathbf{\bar{H}_\text{1} f}\right\|^2  + \left(a_2^2+b_2^2 \right)b_1^2 N + \lambda^2b_0^2 \right) \right),
	\end{equation}
	where $a_i=\sqrt{\frac{K_i}{K_i+1}}$, $b_i =\sqrt{\frac{1}{K_i+1}} $, $i=0,1,2$.
\end{proposition}

\begin{proof}
	See the Appendix \ref{Appendix1}.
\end{proof}

Armed with the concise ergodic capacity upper bound given in Proposition \ref{proposition}, we are ready to move forward to the design of the transmit beam and phase shift beam.
Specifically, we have the following optimization problem P2:
\begin{equation}
\label{p2}
\begin{aligned}
\max_{\mathbf{f},  {\bm \phi}} \quad & C_{\text{up}} \\
\mathrm{s.t.} \quad & \|\mathbf{f}\|^2 = 1 \\
&  | {\phi}_i| = 1, i=1, \dots, N.
\end{aligned}
\end{equation}

Next, depending on whether a LOS path exists for the direct channel between the AP and the user, we give a separate treatment for the two cases, namely, Rician fading and Rayleigh fading.

\subsection{ Rician Fading}
\label{section phi}
When ${\bf g}$ undergoes Rician fading, to maximize the ergodic capacity upper bound, we have the following optimization problem P3:
\begin{equation}
\label{p3}
\begin{aligned}
\max_{\mathbf{f},  {\bm \phi}} \quad & \left|\left(a_2 a_1 \mathbf{\bar{h}}_2^T \boldsymbol{\Phi} \mathbf{\bar{H}}_1 + \lambda a_0 \mathbf{\bar{g}}^T\right) \mathbf{f}\right|^2
+ b_2^2 a_1^2 \left\|\mathbf{\bar{H}_\text{1} f}\right\|^2 \\
\mathrm{s.t.} \quad &  | {\phi}_i| = 1, i=1, \dots, N.
\end{aligned}
\end{equation}

Since the above optimization problem is non-convex, the optimal ${\bf f}$ and ${\bm \phi}$ are difficult to obtain. As such, instead of looking for the optimal solution, we propose an alternating method to obtain a suboptimal solution.
\subsubsection{Design of the phase shift beam}
For given transmit beam $\mathbf{f}$, the optimization problem P3 becomes
\begin{equation} \label{p3_sub1}
\begin{aligned}
\max_{ {\bm \phi}} \quad & \left|a_2 a_1 \mathbf{\bar{h}}_2^T \boldsymbol{\Phi} \mathbf{\bar{H}}_1 \mathbf{f} + \lambda a_0 \mathbf{\bar{g}}^T \mathbf{f}\right|^2 \\
\mathrm{s.t.} \quad & |\phi_i| = 1, i=1, \dots, N.
\end{aligned}
\end{equation}

Recall that ${\bm \Phi}=\text{diag} ({\bm \phi})$. The optimization problem P3 can be rewritten as
\begin{equation}
\begin{aligned}
\max_{ {\bm \phi}} \quad & \left|  a_2 a_1 {\bm \phi}^T \text{diag} (\mathbf{\bar{h}}_2 )  \mathbf{\bar{H}}_1 \mathbf{f} + \lambda a_0 \mathbf{\bar{g}}^T \mathbf{f}\right|^2 \\
\mathrm{s.t.} \quad & |{\phi}_i| = 1, i=1, \dots, N.
\end{aligned}
\end{equation}

Although the above problem is non-convex,  a closed-form solution  can be obtained by exploiting the special structure of the objective function. Specifically, we have the following inequality
\begin{align} \label{E1}
 \left|  a_2 a_1 {\bm \phi}^T \text{diag} (\mathbf{\bar{h}}_2 )  \mathbf{\bar{H}}_1 \mathbf{f} + \lambda a_0 \mathbf{\bar{g}}^T \mathbf{f}\right|^2
 \le
 \left|  a_2 a_1 {\bm \phi}^T \text{diag} (\mathbf{\bar{h}}_2 )  \mathbf{\bar{H}}_1 \mathbf{f}  \right|^2
 +\left|  \lambda a_0 \mathbf{\bar{g}}^T \mathbf{f}\right|^2,
\end{align}
where the equality holds if and only if $\mathsf{angle} \left({\bm \phi}^T \text{diag} (\mathbf{\bar{h}}_2 )  \mathbf{\bar{H}}_1 \mathbf{f} \right)=
\mathsf{angle} \left( \mathbf{\bar{g}}^T \mathbf{f} \right)$.

Next, we will show that a solution, satisfying (\ref{E1}) with equality as well as the phase shift constraints, always exists. With the equality in (\ref{E1}) holding, the optimization problem P3 is equivalent to
\begin{equation}
\begin{aligned}
\max_{{\bm \phi}} \quad & \left| {\bm \phi}^T \text{diag} (\mathbf{\bar{h}}_2 )  \mathbf{\bar{H}}_1 \mathbf{f} \right|^2 \\
\mathrm{s.t.} \quad & |{\phi}_i| = 1, i=1, \dots, N \\
& \mathsf{angle} \left({\bm \phi}^T \text{diag} (\mathbf{\bar{h}}_2 )  \mathbf{\bar{H}}_1 \mathbf{f} \right)=
\mathsf{angle} \left( \mathbf{\bar{g}}^T \mathbf{f} \right).
\end{aligned}
\end{equation}

It is obvious that the optimal solution is given by
\begin{align} \label{rician_op_phi}
 {\bm \phi}^\star=e^{j \mathsf{angle} \left( \mathbf{\bar{g}}^T \mathbf{f} \right)- j\mathsf{angle} \left( \text{diag} (\mathbf{\bar{h}}_2 )  \mathbf{\bar{H}}_1 \mathbf{f} \right) }.
\end{align}

\subsubsection{Design of the transmit beam}
\label{section f}
For given phase shift beam ${\bm \phi}$, problem \eqref{p3} becomes
\begin{equation}\label{p3_sub2}
\begin{aligned}
\max_{\mathbf{f}} \quad & \left|\left(a_2 a_1 \mathbf{\bar{h}}_2^T \boldsymbol{\Phi} \mathbf{\bar{H}}_1 + \lambda a_0 \mathbf{\bar{g}}^T\right) \mathbf{f}\right|^2
+ b_2^2 a_1^2 \left\|\mathbf{\bar{H}_\text{1} f}\right\|^2 \\
\mathrm{s.t.} \quad & \|\mathbf{f}\|^2 = 1,
\end{aligned}
\end{equation}
which can be rewritten  in a more compact form:
\begin{equation}
\begin{aligned}
\max_{\mathbf{f}} \quad & \left\|\mathbf{H f} \right\|^2 \\
\mathrm{s.t.} \quad & \|\mathbf{f}\|^2 = 1,
\end{aligned}
\end{equation}
where
\begin{align*}
\mathbf{H} \triangleq
\begin{bmatrix}
a_2 a_1 \mathbf{\bar{h}}_2^T \boldsymbol{\Phi} \mathbf{\bar{H}}_1 + \lambda a_0 \mathbf{\bar{g}}^T \\
b_2 a_1 \mathbf{\bar{H}_\text{1}}
\end{bmatrix}.
\end{align*}
This is the well-known matrix induced 2-norm problem, which can be solved by performing a singular value decomposition (SVD) of $\mathbf{H}$. Let $\mathbf{H = U \Sigma V^\mathrm{H}}$(the singular values in $\mathbf{\Sigma}$ are in descending order), and then we obtain the optimal beam vector as
\begin{equation}\label{rician_op_f}
\mathbf{f}^\star = {\bf v}_1,
\end{equation}
where ${\bf v}_1$ is the first column of ${\bf V}$.

\subsubsection{Joint Design of the transmit beam and phase shift beam}
The optimization problem (\ref{p3}) is tackled  by solving two sub-problems (\ref{p3_sub1}) and (\ref{p3_sub2}) in an iterative manner, and the details  are summarized in Algorithm \ref{Alg}.

 \begin{algorithm}[!t]
\caption{Alternating Optimization Algorithm}
\label{Alg}
\begin{algorithmic}[1]
\STATE{$\mathbf{Initialization:}$ Given feasible initial solutions ${\bm \phi}_{0}$, ${\bf f}_{0}$ and the iteration index $i=0$.}

\REPEAT
\STATE For given transmit beam ${\bf f}_{i}$,  calculate the optimal phase shift beam accoring to (\ref{rician_op_phi}), which yields ${\bm \phi}_{i+1}$ .
\STATE For given phase shift beam ${\bm \phi}_{i+1}$, compute the optimal transmit beam according to  (\ref{rician_op_f}), which yields ${\bf f}_{i+1}$.
\STATE $i \leftarrow i+1$.

\UNTIL The fractional increase of the objective value  is below a threshold $\varepsilon > 0$.

\STATE $\mathbf{Output:}$ ${\bm \phi}^{\star} ={\bm \phi}_{i}$ and ${\bf f}^{\star}={\bf f}_{i}$.
  \end{algorithmic}
\end{algorithm}

\begin{remark} \label{remark1}
For each sub-problem, the optimal solution is obtained which ensures the objective value of the problem P3 is non-decreasing over iterations. Thus,
 the Algorithm \ref{Alg} is guaranteed to converge.
\end{remark}

\subsection{Rayleigh Fading}

When ${\bf g}$ undergoes Rayleigh fading, i.e., $K_0 = 0$, the ergodic capacity upper bound becomes
\begin{equation} \label{C_up_ray}
C_\mathrm{up} = \log_2\left(1 + \gamma_0 \left(  \left|a_2 a_1 \mathbf{\bar{h}}_2^T \boldsymbol{\Phi} \mathbf{\bar{H}}_1 \mathbf{f}\right|^2
+ b_2^2 a_1^2 \left\|\mathbf{\bar{H}_\text{1} f}\right\|^2   + \left(a_2^2+b_2^2 \right)b_1^2 N + \lambda^2\right) \right).
\end{equation}

Then, we have the following equivalent optimization problem P4:
\begin{equation}
\label{p4}
\begin{aligned}
\max_{\mathbf{f},  {\bm \phi}} \quad & a_2^2 a_1^2 \left|\mathbf{\bar{h}}_2^T \boldsymbol{\Phi} \mathbf{\bar{H}}_1 \mathbf{f}\right|^2
+ b_2^2 a_1^2 \left\|\mathbf{\bar{H}_\text{1} f}\right\|^2 \\
\mathrm{s.t.} \quad & \|\mathbf{f}\|^2 = 1 \\
&  | {\phi}_i| = 1, i=1, \dots, N.
\end{aligned}
\end{equation}

The objective function can be further expressed as
\begin{align}
 a_2^2 a_1^2 \left|\mathbf{\bar{h}}_2^T \boldsymbol{\Phi} \mathbf{\bar{H}}_1 \mathbf{f}\right|^2+ b_2^2 a_1^2 \left\|\mathbf{\bar{H}_\text{1} f}\right\|^2
 = \left\{a_2^2 a_1^2 f_1( {\bm \phi}) + b_2^2 a_1^2 N \right\}f_2(\mathbf{f}),
\end{align}
where
\begin{align}
&f_1( {\bm \phi})= | {\bm \phi}^T \text{diag} \left( \mathbf{a}_N(\theta_\text{AoD,2})  \right) \mathbf{a}_N(\theta_\text{AoA,1})|^2,\\
&f_2(\mathbf{f})=|\mathbf{a}_M^T(\theta_\text{AoD,1}) \mathbf{f}|^2.
\end{align}

As such, the optimization problem P4 becomes
\begin{equation} \label{p4_1}
\begin{aligned}
\max_{\mathbf{f},  {\bm \phi}} \quad & \left\{a_2^2 a_1^2 f_1(\boldsymbol{\phi}) + b_2^2 a_1^2 N \right\}f_2(\mathbf{f}) \\
\mathrm{s.t.} \quad & \|\mathbf{f}\|^2 = 1 \\
& |\phi_i|=1, i=1,\ldots,N.
\end{aligned}
\end{equation}

Noticing that the optimization of ${\bm \phi}$ and $\mathbf{f}$ is decoupled, problem (\ref{p4_1}) can be equivalently transformed into two  sub-problems with respect to ${\bm \phi}$ and ${{\bf f}}$ respectively:
\begin{equation} \label{p4_2}
\begin{aligned}
\max_{   {\bm \phi}} \quad &   f_1(\boldsymbol{\phi})  =|\boldsymbol{\phi}^T \text{diag} \left( \mathbf{a}_N(\theta_\text{AoD,2})  \right) \mathbf{a}_N(\theta_\text{AoA,1})|^2  \\
\mathrm{s.t.} \quad &  |\phi_i|=1, i=1,\ldots,N,
\end{aligned}
\end{equation}
and
\begin{equation} \label{p4_3}
\begin{aligned}
\max_{\mathbf{f}} \quad & f_2(\mathbf{f})=|\mathbf{a}_M^T(\theta_\text{AoD,1}) \mathbf{f}|^2 \\
\mathrm{s.t.} \quad & \|\mathbf{f}\|^2 = 1.
\end{aligned}
\end{equation}

To this end, it is not difficult to show that the optimal solutions for sub-problems (\ref{p4_2}) and (\ref{p4_3}) can be respectively obtained as
\begin{align}
 &{\bm \phi}^{\star}
 =\left( \text{diag} \left( \mathbf{a}_N(\theta_\text{AoD,2})  \right) \mathbf{a}_N(\theta_\text{AoA,1}) \right)^*,
\label{op_phi} \\
 & {\bf f}^\star=\frac{\mathbf{a}_M^*(\theta_\text{AoD,1}) }{\| \mathbf{a}_M(\theta_\text{AoD,1}) \|}=
 \frac{\mathbf{a}_M^*(\theta_\text{AoD,1}) }{\sqrt{M}}.\label{op_f}
\end{align}

\begin{remark}
The optimal phase shift beam is determined by both the AOA of the AP-IRS channel and the AOD
of the IRS-user channel, while the optimal transmit beam is only determined by  the AOD of the AP-IRS channel. In particular, it is desired to align the transmit beam to the IRS direction.
This is reasonable because the transmitter only has  statistical CSI about the IRS-aided channels.
\end{remark}

\begin{proposition} \label{prop2}
    With the optimal transmit beamformer ${\bf f}^\star$ and phase shift beam ${\bm \phi}^\star$, the ergodic capacity upper bound is given by
    \begin{align}
        C_{\text{up}}=\log_2\left(1 + \gamma_0 \left( a_2^2 a_1^2 M N^2 + b_2^2 a_1^2 M N + (a_2^2+b_2^2) b_1^2 N + \lambda^2 \right) \right)
    \end{align}
\end{proposition}
\begin{proof}
Substituting (\ref{op_phi}) and (\ref{op_f}) into (\ref{C_up_ray}), we can obtain the desired result.
\end{proof}

Proposition \ref{prop2} shows that by applying the proposed beamforming scheme, a power gain of order $N^2 M$  corresponding to the AP-IRS-user link can be obtained, although only statistical CSI is available at the AP. The $M$-fold gain comes from transmit beamforming, while the $N^2$-fold gain is achieved due to
both the phase shift beamforming and the inherent aperture gain of the IRS.
It is worth noting that the proposed statistical CSI based beamforming scheme achieves the
same power gain as the beamforming scheme with full CSI \cite{wu2019intelligent}.

\section{Numerical Results} \label{S4}
In this section, numerical results are presented to validate the performance of our proposed beamforming schemes. Unless otherwise specified, the following parameters are used in the simulations. The signal bandwidth is set to be 180 kHz and the noise power spectral density is -170 dBm/Hz. The transmit power $P$ is set to be -40 dBm. The distances are $d_0 = 200 \ \mathrm{m}, d_1  = 250 \ \mathrm{m},  d_2 = 50 \ \mathrm{m}$, with the pathloss exponents $\alpha_0 = 3.5$, $\alpha_1 =2.5$ and $\alpha_2 = 2.2$. Moreover, the Rician factors are normalized as $a_i =  \sqrt{\frac{K_i}{K_i+1}}, b_i = \sqrt{\frac{1}{K_i+1}}, K_i = 1$ for $i = 0, 1, 2$. The numbers of AP antennas and reflecting elements are set to be $M=8$ and $N=128$, respectively. The AoA and AoD in the LoS components are randomly chosen from $[0, 2\pi)$. The Monte Carlo results are averaged over 10,000 independent trials.

Fig.~\ref{upperbound_tightness} demonstrates the tightness of the upper bound approximation of the ergodic capacity in Proposition~\ref{proposition}, where the transmit beam and phase shift beam are designed according to Algorithm \ref{Alg}.
We can see that the upper bound approximation curves almost overlap with the Monte Carlo curves, which comfirms the tightness of the upper bound. Furthermore, we observe the intuitive result that the ergodic capacity increases significantly with $N$, indicating the benefit of applying a large number of reflecting elements.
In addition, as the Rician K-factor increases, the ergodic capacity becomes larger, due to the fact that only statistical CSI is available and used for the beamforming design.

\begin{figure}[htbp]
	\centering
	\includegraphics[width=3.5in]{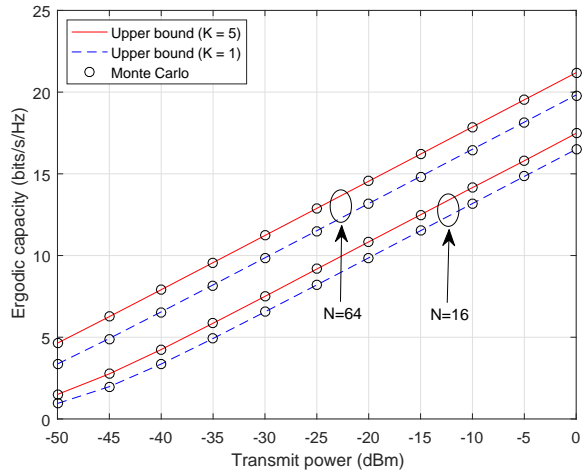}
	\caption{Comparison of upper bound approximation results and Monte Carlo results}
	\label{upperbound_tightness}
\end{figure}

Fig.~\ref{fig_converge} depicts the convergence of Algorithm \ref{Alg} with different numbers of reflecting elements and AP antennas, where either an optimization of the transmit beamformer or phase shift beam is considered as an iteration.
As can be readily observed, for any configuration of $M$ and $N$, the objective function is non-decreasing after each iteration, which is consistent with our analysis in Remark \ref{remark1}. Moreover, the proposed algorithm converges very fast. In all configurations, the objective function reaches a sub-optimal point after only 3 iterations.
\begin{figure}[htbp]
	\centering
	\includegraphics[width=3.5in]{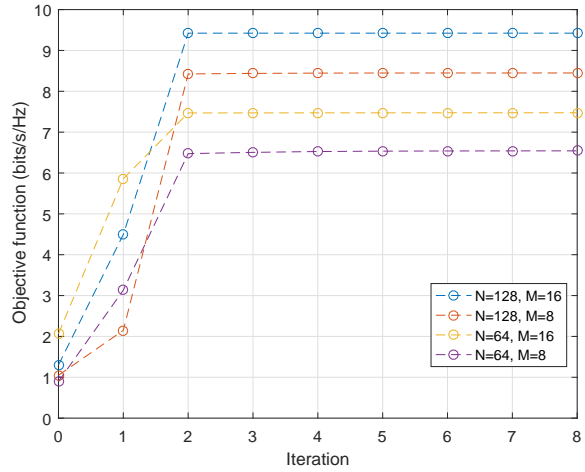}
	\caption{Convergence of Algorithm \ref{Alg}}
	\label{fig_converge}
\end{figure}

Fig. \ref{fig_Rician} shows the performance of the proposed algorithm in  the Rician fading case, i.e., Algorithm \ref{Alg}.
For comparison, both the SDR  algorithm in \cite{wu2019intelligent} and the random scheme
are presented as the benchmark.
As can be seen, the proposed algorithm performs much better than the random scheme.
In addition, our proposed algorithm achieves similar performance to the benchmark algorithm assuming full instantaneous CSI, although only statistical CSI is exploited.
Moreover, as the Rician K-factor increases, the gap between
these two algorithms becomes very small. This is because with  a large Rician K-factor, the CSI is dominated by the LOS channels.
\begin{figure}[htbp]
	\centering
	\includegraphics[width=3.5in]{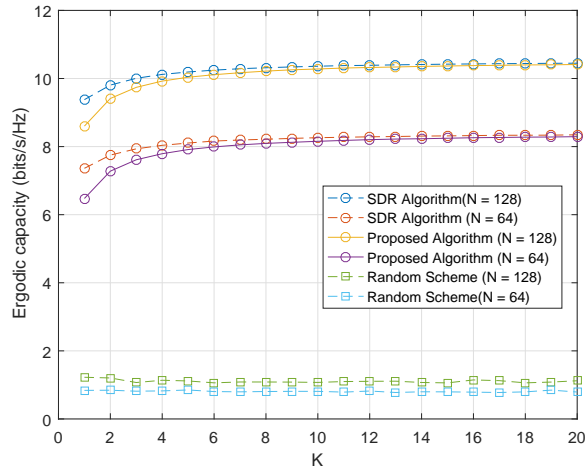}
	\caption{Performance of the proposed algorithm in the Rician fading case}
	\label{fig_Rician}
\end{figure}

Fig. \ref{fig_Rayleigh} shows the performance of the proposed beamforming scheme when the transmitter-receiver channel undergoes Rayleigh fading, where the analytical results are generated according to Proposition \ref{prop2}. For comparison, both the SDR  algorithm in \cite{wu2019intelligent} and the random scheme are presented as the benchmark.
As can be readily seen, the analytical results match well with the Monte Carlo simulation results, thereby validating the correctness of Proposition \ref{prop2}. In addition, our proposed algorithm is superior to the random scheme.
Moreover, the performance of the proposed beamfoming scheme is very close to that of the benchmark algorithm assuming  instantaneous CSI. As the Rician K-factor or the number of reflecting elemnets increases, the gap between these two schemes  becomes smaller.
It is worth noting that the benchmark algorithm operates in an iterative manner,
while our proposed beamfoming scheme has a closed-form solution, thereby having much lower complexity.
\begin{figure}[htbp]
	\centering
	\includegraphics[width=3.5in]{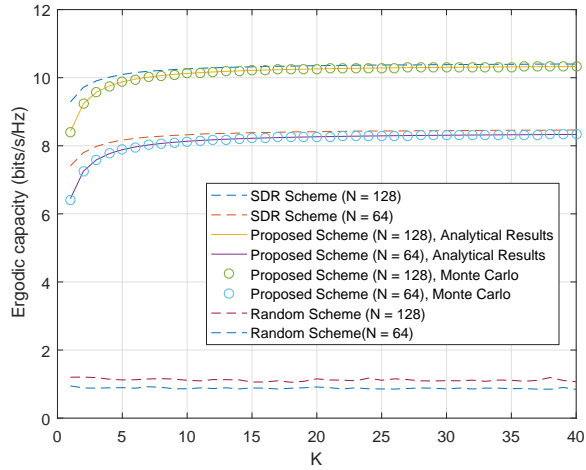}
	\caption{Performance of the proposed beamforming scheme in the Rayleigh fading case}
	\label{fig_Rayleigh}
\end{figure}

Fig. \ref{Rician_VS_Rayleigh} compares the two different fading cases.
As can be readily observed, the Rician fading case is superior to the Rayleigh fading case.
This is because the transmitter does not have any CSI about the direct channel in the Rayleigh fading case, but has some statistical CSI about the direct channel in the Rician fading case, due to the existence of the LOS path.
As the number of reflecting elements becomes larger (for example, above 50), the two fading cases  achieve similar performance. The reason is that with a large number of reflecting elements, the performance of both cases is dominated by the strong IRS-aided  channel.
\begin{figure}[htbp]
	\centering
	\includegraphics[width=3.5in]{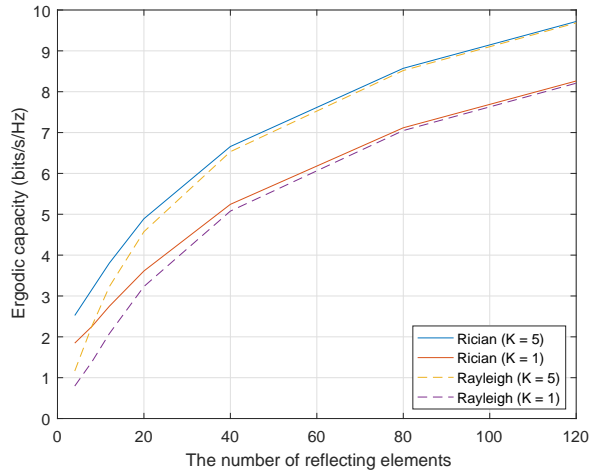}
	\caption{Comparison of the Rician fading case and the Rayleigh fading case}
	\label{Rician_VS_Rayleigh}
\end{figure}

Fig. \ref{Rician_VS_Rayleigh_SNR} shows the impact of transmit power on the capacity in two different fading cases.
It is intuitive that the performance of both cases improve as the transmit power increases. Moreover, when the number of reflecting elements is small, there is a noticeable capacity gap between the two fading cases. In contrast, with a large number of reflecting elements, the two cases achieve almost the same performance, mainly due to the dominated IRS-aided channel.
\begin{figure}[htbp]
	\centering
	\includegraphics[width=3.5in]{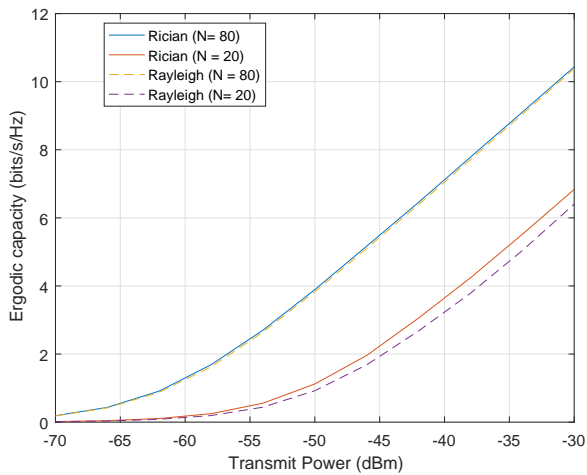}
	\caption{The impact of transmit power on the ergodic capacity}
	\label{Rician_VS_Rayleigh_SNR}
\end{figure}

\section{Conclusion} \label{S5}
This paper has addressed the problem of joint transmit beam and phase shifts design in an IRS-aided MISO communication system. Assuming that only statistical CSI is available at the transmitter, efficient algorithms are designed to maximize the ergodic capacity of the system. Specifically, for the Rician fading case, an iterative algorithm is proposed and the convergence of the algorithm is established. For the Rayleigh fading case, closed-form designs are obtained. Simulation results show that the proposed algorithm achieves similar performance as the algorithm requiring instantaneous CSI.




\begin{appendix} \label{Appendix1}
\section{{Proof of Proposition~\ref{proposition}}}
Applying the Jensen's inequality, we have
\begin{align}
 \mathbb{E}\left\{ \log_2 \left( 1+\gamma \right) \right\}
\lessapprox   \log_2 \left(1+\mathbb{E}\{\gamma\} \right)
= \log_2 \left( 1+\gamma_0 \mathbb{E} \left\{ \left|(\mathbf{h}_2^T \boldsymbol{\Phi} \mathbf{H}_1+\lambda\mathbf{g}^T) \mathbf{f} \right|^2 \right\} \right).
\end{align}
The remaing task is the derivation of $\mathbb{E} \left\{ \left|(\mathbf{h}_2 \boldsymbol{\Phi} \mathbf{H}_1+\lambda\mathbf{g}) \mathbf{f} \right|^2 \right\}$. Applying the binomial expansion theorem, we have
	\begin{align}
	& \left| \left(\mathbf{h}_2^T \boldsymbol{\Phi} \mathbf{H}_1+\lambda \mathbf{g}^T \right) \mathbf{f} \right|^2 \\
	=& \left| \left(a_2 \mathbf{\bar{h}}_2^T+ b_2 \mathbf{\tilde{h}}_2^T \right) \boldsymbol{\Phi} \left(a_1 \mathbf{\bar{H}}_1 + b_1 \mathbf{\tilde{H}}_1 \right) \mathbf{f} + \lambda\left(a_0 \mathbf{\bar{g}}^T + b_0 \mathbf{\tilde{g}}^T \right) \mathbf{f} \right|^2 \nonumber \\
	=& \left| x_1 + x_2 + x_3 + x_4 + x_5 \right|^2,\nonumber
	\end{align}
where
\begin{align}
& x_1 = (a_2 a_1 \mathbf{\bar{h}}_2^T \boldsymbol{\Phi} \mathbf{\bar{H}}_1+ \lambda a_0 \mathbf{\bar{g}}^T) \mathbf{f}, \\
& x_2 = a_2 b_1 \mathbf{\bar{h}}_2^T \boldsymbol{\Phi} \mathbf{\tilde{H}}_1 \mathbf{f}, \\
& x_3 = b_2 a_1 \mathbf{\tilde{h}}_2^T \boldsymbol{\Phi} \mathbf{\bar{H}}_1 \mathbf{f}, \\
& x_4 = b_2 b_1 \mathbf{\tilde{h}}_2^T \boldsymbol{\Phi} \mathbf{\tilde{H}}_1 \mathbf{f}, \\
& x_5 = \lambda b_0 \mathbf{\tilde{g}}^T \mathbf{f}.
\end{align}

It is easy to observe that $x_1$ is a constant and $\mathbb{E}\{x_i\} = 0$ holds for $i = 2,3,4,5$. Besides, since $\mathbf{\tilde{h}}_2$, $\mathbf{\tilde{H}}_1$ and $\mathbf{\tilde{g}}$ have zero means and are indepedent with each other, we can derive that
\begin{align}
& \mathbb{E} \left\{ \left|(\mathbf{h}_2^T \boldsymbol{\Phi} \mathbf{H}_1+\lambda \mathbf{g}^T) \mathbf{f} \right|^2 \right\} \\
= & \mathbb{E} \left\{ \left| x_1 + x_2 + x_3 + x_4 + x_5 \right|^2 \right\} \nonumber \\
= & |x_1|^2 + \mathbb{E}\{|x_2|^2\} + \mathbb{E}\{|x_3|^2\} + \mathbb{E}\{|x_4|^2\} + \mathbb{E}\{|x_5|^2\}.\nonumber
\end{align}

Denote ${\bf w} \triangleq\tilde{\bf H}_1 {\bf f}$. $\mathbb{E}\{|x_2|^2\} \in \mathbb{C}^{N \times 1}$ can be expressed as
\begin{align}
\mathbb{E}\{|x_2|^2\}= a_2^2 b_1^2
\mathbb{E}
\left\{
\mathsf{tr}
\left(
\bar{\bf h}_{2}^{*} \bar{\bf h}_{2}^{T} {\bf w} {\bf w}^H
\right)
\right\}=a_2^2 b_1^2
\mathsf{tr}
\left(
\bar{\bf h}_{2}^{*} \bar{\bf h}_{2}^{T} \mathbb{E}
\left\{ {\bf w} {\bf w}^H \right\}
\right).
\end{align}
Noticing that $\mathbb{E}
\left\{ {\bf w} {\bf w}^H \right\}={\bf I}_N$, we have
\begin{align}
    \mathbb{E}\{|x_2|^2\}=a_2^2 b_1^2\mathsf{tr}
\left(
\bar{\bf h}_{2}^{*} \bar{\bf h}_{2}^{T}
 {\bf I}_N
\right)=a_2^2 b_1^2 N.
\end{align}

Following the similar lines, it can be derived that
\begin{align}
&\mathbb{E}\{|x_3|^2\} =
b_2^2 a_1^2
\left\|\mathbf{\bar{H}_\text{1} f}\right\|^2, \\
&\mathbb{E}\{|x_4|^2\} =
b_2^2 b_1^2 N, \\
&\mathbb{E}\{|x_5|^2\} = \lambda^2 b_0^2.
\end{align}

Summing over all the values yields the desired result.

\end{appendix}

\end{document}